\documentclass[a4paper,UKenglish,cleveref, autoref, thm-restate]{lipics-v2021}

\usepackage{amsmath, amssymb, amsfonts}
\usepackage{todonotes}
\usepackage{amsmath,amsfonts,amssymb,amsthm}
\usepackage{graphicx,color}
\usepackage{boxedminipage}
\usepackage{algorithm}
\usepackage[noEnd=false]{algpseudocodex}
\usepackage{framed}
\usepackage{xfrac}
\usepackage{xspace}
\usepackage{todonotes}
\usepackage{footnote}
\usepackage{mathtools}
\usepackage{mathrsfs}
\usepackage{todonotes}
\usepackage{footnote}
\usepackage{rotating}
\usepackage{nicefrac}
\newif\iffull
\fulltrue

\usepackage{thmtools}
\usepackage{thm-restate}

	\usepackage{mdframed}
	
	\usepackage{tcolorbox}
	\tcbuselibrary{skins}
	\usepackage{xcolor}
	
	\colorlet{mix}{red!50!black}
	
	\usepackage{hyperref}

	\usepackage{xspace}

	\makeatletter
	\newcommand*{\rom}[1]{\expandafter\@slowromancap\romannumeral #1@}
	\makeatother

	\usepackage{mathtools}

	\newcommand{\fpt} {{\sf FPT}\xspace}

	\newcommand{\dd}{\mathbf{d}}

	\newcommand{\restatemarker}{\textup{[$\spadesuit$]} \ignorespaces}
	\AddToHook{cmd/appendix/after}{\renewcommand{\restatemarker}{}}
	
	\usepackage{xifthen}
	\usepackage{tabularx}
	\usepackage{bm}
	\usepackage{tikz} 
	\usetikzlibrary{calc}
	\usepackage{thm-autoref}

	\crefname{lin}{line}{lines}
	\Crefname{lin}{Line}{Lines}

	\DeclareMathOperator{\operatorClassNP}{NP}
	\newcommand{\classNP}{\ensuremath{\operatorClassNP}}

	\newcommand{\Oh}{\mathcal{O}}
	\DeclareMathOperator{\poly}{poly}
	\newcommand{\fptas}{{\sf FPT-AS}\xspace}

	\newcommand{\ball}{{\sf ball}}
	\newcommand\numberthis{\addtocounter{equation}{1}\tag{\theequation}}
	\newcommand{\wcost}{{\sf wcost}}
	
	\newcommand{\sucprob}{$\exp\lr{ -O\lr{\frac{k}{\varepsilon} \log(\frac{k}{\varepsilon}) \lambda(\frac{\varepsilon}{40}) }}$}
	
	\newcommand{\ballintalg}{\ensuremath{\mathcal{C}}\xspace}
	\newcommand{\ballint}{{\sc Ball Intersection}\xspace}
	\newcommand{\kmeans}{$k$-\textsc{Means}\xspace}
	
	\newlength{\RoundedBoxWidth}
	\newsavebox{\GrayRoundedBox}
	\newenvironment{GrayBox}[1]%
	{\setlength{\RoundedBoxWidth}{.93\textwidth}
		\def\boxheading{#1}
		\begin{lrbox}{\GrayRoundedBox}
			\begin{minipage}{\RoundedBoxWidth}}%
			{   \end{minipage}
		\end{lrbox}
		\begin{center}
			\begin{tikzpicture}%
				\node(Text)[draw=black!20,fill=white,rounded corners,%
				inner sep=2ex,text width=\RoundedBoxWidth]%
				{\usebox{\GrayRoundedBox}};
				\coordinate(x) at (current bounding box.north west);
				\node [draw=white,rectangle,inner sep=3pt,anchor=north west,fill=white] 
				at ($(x)+(6pt,.75em)$) {\boxheading};
			\end{tikzpicture}
	\end{center}}     
	
	\newenvironment{defproblemx}[2][]{\noindent\ignorespaces%
		\FrameSep=6pt%
		\parindent=0pt%
		\vspace*{-1.5em}
		\ifthenelse{\isempty{#1}}{%
			\begin{GrayBox}{\textsc{#2}}%
			}{%
				\begin{GrayBox}{\textsc{#2}  parameterized by~{#1}}%
				}
				\begin{tabular*}{\textwidth}{@{\hspace{.1em}} >{\itshape} p{1.8cm} p{0.8\textwidth} @{}}%
				}{
				\end{tabular*}%
			\end{GrayBox}%
			\ignorespacesafterend
		}  
		
		\DeclarePairedDelimiter{\LRb}{\lbrace}{\rbrace}
		\newcommand{\LR}{\LRb*}
		\DeclarePairedDelimiter{\lrb}{(}{)}
		\newcommand{\lr}{\lrb*}

		\newcommand{\cI}{\mathcal{I}}

		\newcommand{\cost}{\mathsf{cost}}

		\newcommand{\OPT}{{\sf OPT}}
		
		\renewcommand{\d}{\mathsf{dist}}
		\newcommand{\fac}{\mathbb{F}}
		\newcommand{\cli}{P}
		
		\newcommand{\real}{\mathbb{R}}
		\newcommand{\kmed}{$k$-\textsc{Median}\xspace}
		\newcommand{\kcenter}{$k$-\textsc{Center}\xspace}
		\newcommand{\probname}{{\sc Hybrid $k$-Clustering}\xspace}
		
		\newtcolorbox{mybox}[1]{colback=white!5!white,colframe=gray!75!black,colbacktitle=white!5!white,coltitle=black!70!black,sharp corners=all,title={#1}}

\title{Dimension-Free Parameterized Approximation Schemes for Hybrid Clustering}

\titlerunning{Dimension-Free Parameterized Approximation Schemes for Hybrid Clustering} 

\author{Ameet Gadekar}
{CISPA Helmholtz Center for Information Security, Saarbr\"{u}cken, Germany}{ameet.gadekar@cispa.de}{https://orcid.org/0009-0004-8040-9881}{Partially supported by the Israel Science Foundation (grant No. 1042/22).}

\author{Tanmay Inamdar}
{Indian Institute of Technology Jodhpur, Jodhpur, India}{taninamdar@gmail.com}{https://orcid.org/0000-0002-0184-5932}{Supported by  IITJ Research Initiation Grant (grant number I/RIG/TNI/20240072).}

\hideLIPIcs

\authorrunning{A. Gadekar and T. Inamdar} 
\Copyright{Ameet Gadekar and Tanmay Inamdar}



\nolinenumbers

\category{} 

\relatedversion{} 

\acknowledgements{Tanmay would like to thank his co-authors from \cite{FominGISZ24}---specifically Fedor V.~Fomin---for formulating Hybrid $k$-Clustering problem and introducing him to it. This work was partially carried out while Ameet was at Bar-Ilan University, Israel.}

\EventEditors{Olaf Beyersdorff, Micha\l{} Pilipczuk, Elaine Pimentel, and Nguyen Kim Thang}
\EventNoEds{4}
\EventLongTitle{42nd International Symposium on Theoretical Aspects of Computer Science (STACS 2025)}
\EventShortTitle{STACS 2025}
\EventAcronym{STACS}
\EventYear{2025}
\EventDate{March 4--7, 2025}
\EventLocation{Jena, Germany}
\EventLogo{}
\SeriesVolume{327}
\ArticleNo{24}

\begin{document}

\maketitle

\begin{abstract}
{\sc Hybrid $k$-Clustering} is a model of clustering that generalizes two of the most widely studied clustering objectives: {\sc $k$-Center} and {\sc $k$-Median}. In this model, given a set of $n$ points $P$, the goal is to find  $k$ centers such that the sum of the \emph{$r$-distances} of each point to its nearest center is minimized.  The $r$-distance between two points $p$ and $q$ is defined as $\max\LR{\d(p, q)-r, 0}$ --- this represents the distance of $p$ to the boundary of the $r$-radius ball around $q$ if $p$ is outside the ball, and $0$ otherwise. This problem was recently introduced by Fomin et al.~[APPROX 2024], who designed a $(1+\varepsilon, 1+\varepsilon)$-bicrtieria approximation that runs in time $2^{(kd/\varepsilon)^{O(1)}} \cdot n^{O(1)}$ for inputs in $\mathbb{R}^d$; such a bicriteria solution uses balls of radius $(1+\varepsilon)r$ instead of $r$, and has a cost at most $1+\varepsilon$ times the cost of an optimal solution using balls of radius $r$.

In this paper we significantly improve upon this result by designing an approximation algorithm with the same bicriteria guarantee, but with running time that is \fpt\ only in $k$ and $\varepsilon$ --- crucially, removing the exponential dependence on the dimension $d$. This  resolves an open question posed in their paper. Our results extend further in several directions. First, our approximation scheme works in a broader class of metric spaces, including  doubling spaces, minor-free, and bounded treewidth metrics. Secondly,  our techniques yield a similar bicriteria {\sf FPT}-approximation schemes for other variants of {\sc Hybrid $k$-Clustering}, such as when the objective features the sum of $z$-th power of the $r$-distances. Finally, we also design a coreset for {\sc Hybrid $k$-Clustering} in doubling spaces, answering another open question from the work of Fomin et al. 
\end{abstract}

\newpage 

\section{Introduction} \label{sec:intro}

\kcenter, \kmed, and \kmeans 
are among the most popular clustering problems both in theory and practice, with numerous applications in areas such as machine learning~\cite{llyod,HS85,AryaGKMMP04,braverman2019coresets,anthony2010plant,bhattacharya2014new,pmlr-v134-makarychev21a,ghadiri2021socially}, 
facility location problems~\cite{10.1287/mnsc.9.4.643,ccc4ba45-b343-377b-8adf-62e4761686a7,DBLP:journals/jacm/JainV01,DBLP:conf/stoc/JainMS02,DBLP:journals/iandc/Li13,DBLP:conf/soda/Cohen-Addad0LS23}, and computational geometry~\cite{BadoiuHI02,KumarSS10}, among others.\footnote{These citations are supposed not to be comprehensive, but representative--an interested reader may use them as a starting point for the relevant literature.\label{foot:extensivecite}} All of these problems have long been known to be \classNP-hard, even in the plane \cite{FederG88,megiddo1984complexity,MahajanNV12}. To cope with these intractability results, there has been several decades of research on designing approximation algorithms for these problems, that have lead to polynomial-time constant-factor approximation algorithms for all of them in general metric spaces \cite{HochbaumM85,Cohen-AddadSS21}. Moreover, for more ``structured'' inputs---such as when the points belong to Euclidean spaces, or planar graph metrics---one can find near-optimal solutions using approximation schemes~\cite{DBLP:journals/siamcomp/FriggstadRS19,DBLP:journals/siamcomp/Cohen-AddadKM19,cohen2019polynomial,Arora,KumarSS10,BadoiuHI02,AbbasiBBCGKMSS24}.\footnote{This includes polynomial-time as well as \fpt~approximation schemes.}
Furthermore, given the simplicity and ubiquity of these vanilla clustering objectives, the techniques used for obtaining these results are subsequently used to find clustering with additional constraints, such as capacities \cite{Cohen-AddadL19Capacitated}, fairness \cite{BandyapadhyayFS24}, and outliers~\cite{KrishnaswamyLS18,AgrawalISX23}.$^{\ref{foot:extensivecite}}$ One recurring theme in the literature has been to define a unified clustering objective that also captures classical clustering problems like \kmed, \kcenter as special cases \cite{Tamir01,ByrkaSS18,ChakrabartyS19Improved}. 

Along this line of research, Fomin et al.~\cite{FominGISZ24} recently introduced \probname. In this problem, we are given an instance $\cI = (M = (\cli, \fac, \d), k, r)$, where $M = (\cli, \fac, \d)$ is a metric space\footnote{Fomin et al.~\cite{FominGISZ24} only considered the special case of Euclidean inputs, i.e., when $\cli \subset \fac = \real^{d}$; however, the underlying problem can be defined for arbitrary metric spaces.} on the set of clients $\cli$ and a set of facilities $\fac$ with distance function $\d$. Here, $k$ is a positive integer denoting the number of clusters, and $r$ is a non-negative real denoting the radius. The goal is to find a set $X \subseteq \fac$ of $k$ centers, such that, $\cost_r(P, X) \coloneqq \sum_{p \in P} \d_r(p, X)$ is minimized---here, for any point $p$, any set $Q \subseteq \cli \cup \fac$, and a real $\alpha$,  $\d_\alpha(p, Q)$ is defined as $\d_\alpha(p, Q)\coloneqq \max\LR{\d(p, Q) - \alpha, 0}$, and $\d(p, Q) \coloneqq \min_{q \in Q} \d(p, q)$. Fomin et al.~\cite{FominGISZ24} proposed several motivations for studying this cost function. Indeed, \probname can be seen as a \emph{shape fitting} problem, where one wants to find the ``best'' set of $k$ balls of radius $r$ that fit the given set of points, where the quality of the solution is determined by the sum of distances of each point to the nearest point on the boundary of the ball---this is analogous to the classical regression, where one wants to find the ``best'' linear function fitting the given set of points. 
\textsc{Projective Clustering} is a well-studied generalization of linear regression, where the aim is to find $k$ affine spaces that minimizes the sum of distances from the points to these spaces (see e.g., \cite{Tukan0ZBF22}), which is also closely related to the model considered in \probname. Furthermore, \cite{FominGISZ24} gave another motivation for studying \probname -- namely, placing $k$ WiFi routers with identical circular coverage, where the clients that lie outside the coverage area need to travel to the boundary of the nearest ball in order to receive coverage. Finally, the name of the problem is motivated from the fact that the objective is a kind of ``hybrid'' between the \kcenter and \kmed costs, and generalizes both of them. Indeed, as observed in \cite{FominGISZ24}, \probname with $r = 0$ is equivalent to \kmed, while, when $r$ is set to be the optimal radius for \kcenter, \probname reduces to \kcenter. These observations immediately rule out \emph{uni-criteria} polynomial-time approximation schemes that violate only the cost, or only the radius by any arbitrary factor $\alpha > 1$ (as formalized in Proposition 1 of \cite{FominGISZ24}). Indeed, such approximation algorithms---even with running times \fpt in $k$---would imply exact \fpt algorithms for \kcenter and \kmed in low-dimensional continuous Euclidean spaces. To our knowledge, such an algorithm for \kmed is not known (nor is admittedly a lower bound), whereas \cite{Marx06} shows W[1]-hardness of \kcenter even in $\mathbb{R}^2$. Therefore, given the current state of the art, a \emph{bicriteria} approximation scheme for \probname is essentially the best outcome, even if one allows a running time that is \fpt in $k$.

For $\alpha, \beta \ge 1$, an \emph{$(\alpha, \beta)$-bicrteria approximation} for \probname is an algorithm that returns a solution $X \subseteq \fac$ of size at most $k$ satisfying $\cost_{\beta r}(P, X) \le \alpha \cdot \OPT_r$. Note that  the bicriteria solution $X$ is allowed to consider $\d_{\beta r}(p, X) = \max\LR{\d(p, X) - \beta r, 0}$ instead of $\d_r(p, X)$ \emph{and} is also allowed to find a solution of cost $\alpha \OPT_r$, where $\OPT_r$ is the cost of an optimal solution w.r.t.~radius $r$, i.e., without any violation of the radius. The main result of Fomin et al.~\cite{FominGISZ24} was an $(1+\varepsilon, 1+\varepsilon)$-bicriteria approximation for inputs in $\real^d$ in time $2^{(kd/\varepsilon)^{O(1)}}\cdot n^{O(1)}$, where $n = |P|$. An exponential dependence on the dimension $d$ appears inevitable using their approach, although it was not clear whether such a dependence is required. Indeed, for Euclidean \kmed and \kcenter, one can obtain \emph{dimension-free} approximation schemes with \fpt dependence only on $k$ and $\varepsilon$~\cite{BadoiuHI02,KumarSS05,AbbasiClustering23}. This naturally motivates the following question, which was explicitly asked as an open question in \cite{FominGISZ24}.
\begin{mdframed}[backgroundcolor=gray!10,topline=false,bottomline=false,leftline=false,rightline=false] 
	\centering
	\textbf{Question 1.} ``An immediate question is whether improving or removing the FPT dependence on the dimension $d$ is possible[...]''
\end{mdframed}
In this work, we answer this question in the affirmative by designing a (randomized) bicriteria FPT Approximation Scheme (\fptas) parameterized by $k$ and $\varepsilon$\footnote{Such algorithms run in time $f(k,\varepsilon)n^{O(1)}$ and output a $(1+\varepsilon,1+\varepsilon)$-bicriteria solution.
	Another term for \fptas is {\sf EPAS}, which stands for Efficient Parameterized Approximation Schemes.} for the (continuous) Euclidean instances of \probname, stated formally in the following theorem.

\begin{restatable}[Bicriteria FPT-AS for Euclidean Spaces]{theorem}{euclideanthm}
	There exists a randomized algorithm, that, given an instance of \probname in $\real^d$ for any dimension $d$,  runs in time $2^{O\lr{k \log k \cdot (\nicefrac{1}{\varepsilon^5}) \log^2 (\nicefrac{1}{\varepsilon})}} \cdot n^{\Oh(1)}$, and returns a $(1+\varepsilon, 1+\varepsilon)$-bicrtieria approximation with high probability.
\end{restatable}

The algorithm of \cite{FominGISZ24} involves a pre-processing step that separately handles ``\kmed-like'', ``\kcenter-like'' instances, using the techniques specific to these respective problems. Then, the main algorithm handles the remaining instances that cannot be directly reduced to the respective problems. This approach somewhat undermines the goal of defining a unified problem that captures both of these problems as special cases. In contrast, our algorithm adopts a uniform approach by exploiting the intrinsic structure of the hybrid problem, without the need to separately handle some instances. 
In fact, 
our algorithm works for a broader class of metric spaces, called metric spaces of \emph{bounded algorithmic scatter dimension}---a notion recently introduced by Abbasi et al.~\cite{AbbasiClustering23}, and further studied by Bourneuf and Pilipczuk~\cite{BourneufP25}. These metric spaces capture several interesting and well-studied metric spaces, see below. 
We give a formal definition of the notion of algorithmic scatter dimension in \Cref{sec:prelims}; however, a good mental picture to keep is to think of them as essentially doubling spaces, i.e., metric spaces with good ``packing-covering'' properties.\footnote{Although this is a good mental picture, it is ultimately inaccurate since the class of metric spaces of bounded algorithmic scatter dimension is strictly larger than that of doubling spaces. Indeed, continuous Euclidean space of high ($\omega(\log n)$) dimension does not have bounded doubling dimension, yet it does have bounded algorithmic scatter dimension.} Our general result is stated below.
\begin{restatable}[Informal version of~\Cref{thm:main}]{theorem}{maintheorem} \label{thm:maintheorem}
	\probname admits a randomized bicriteria \fptas in metrics of bounded algorithmic $\varepsilon$-scatter dimension, parameterized by $k$ and $\varepsilon$. In particular, \probname admits randomized bicriteria \fptas parameterized by $k$ and $\varepsilon$, in continuous Euclidean spaces of any dimension, metrics of bounded doubling dimension, bounded treewidth metrics, and metrics induced by graphs from any fixed proper minor-closed graph class.
\end{restatable}
We give a technical overview of this result in \Cref{ss:overview}, and describe the approximation algorithm and its analysis in \Cref{sec:algorithm}. 

In their work, \cite{FominGISZ24} also defined {\sc Hybrid $(k, z)$-Clustering} problem, which features the $z$-th power of $r$-distances, i.e., the objective is to minimize $\cost_{r}(P, X, z) \coloneqq \sum_{p \in P} (\d_r(p, X))^z$, where $z \ge 1$. They designed a bicriteria \fptas with similar running time for {\sc Hybrid $(k, 2)$-Clustering}, i.e., a hybrid of \kmeans and \kcenter; but left the possibility of obtaining a similar result for the general case of $z \ge 1$ conditional on the existence of a certain sampling-based approximation algorithm for the $(k, z)$-clustering problem for Euclidean inputs. Using our approach, we can obtain a bicriteria \fptas for any fixed $z \ge 1$ in a unified way, whose running time is independent of the dimension $d$. In fact, our approach works for a much more general problem of {\sc Hybrid Norm $k$-Clustering}. We discuss these extensions in \Cref{sec:extensions}. 

Next, we turn to another open direction mentioned in the work of Fomin et al.~\cite{FominGISZ24}.
\begin{mdframed}[backgroundcolor=gray!10,topline=false,bottomline=false,leftline=false,rightline=false] 
	\textbf{Question 2.} ``Another intriguing question is the design
	of coresets for \probname, which could also have some implications for [\textbf{Question 1}].''
\end{mdframed}
In this paper, we also answer this question affirmatively by designing coresets for \probname in metric spaces of bounded doubling dimension. Specifically, we prove the following result.
\begin{restatable}[Coreset for \probname]{theorem}{coresettheorem}\label{thm:coresetthm}
	There exists an algorithm that takes as input an instance $\cI = ((\cli, \fac, \d), k, r)$ of \probname in doubling metric of dimension $d$ and a parameter $\varepsilon \in (0, 1)$, and  in time  $2^{O(d\log(1/\varepsilon))} |\cI|^{O(1)}$ returns a pair $(\cli', w)$, where $\cli' \subseteq \cli$ has size $2^{O(d \log(1/\varepsilon)} k \log |P|$, and $w: \cli' \to \mathbb{N}$, such that the following property is satisfied for any $X \subseteq \fac$ of size at most $k$:
	$$\left| \wcost_r(\cli', X) - \cost_r(\cli, X) \right| \le \varepsilon \cost_r(\cli, X)$$
	Here, $\wcost_r(\cli, X) \coloneqq \sum_{p \in \cli'} w(p) \cdot \d_r(p, X)$. 
\end{restatable}

\subsection{Technical Overview}\label{ss:overview}

In this section, we highlight our technical and conceptual contributions of our main result (\Cref{thm:maintheorem}).

A natural direction to design dimension-free parameterized approximation algorithm for \probname is to leverage insights from the framework of Abbasi et al.~\cite{AbbasiClustering23}, who designed dimension-free parameterized approximations for wide range of clustering problems. 
However, in \probname, the objective is a function of $r$-distances, instead of true distances, as in~\cite{AbbasiClustering23}. This introduces several technical challenges. For instance, a most obvious roadblock is that the $r$-distances do not satisfy triangle inequality. Moreover, many of the ``nice'' properties that are enjoyed by the true distances in ``structured'' metric spaces (e.g., covering-packing property in doubling spaces, or properties of shortest paths in sparse graphs) do not extend readily to the $r$-distances. Therefore, we have to overcome several technical and conceptual challenges in order to leverage the ideas presented in the framework of \cite{AbbasiClustering23}. Let us first recall the key notion from this paper that is crucially used in this framework.\footnote{In fact, the algorithm of our paper and \cite{AbbasiClustering23} works for a weaker notion called algorithmic scatter dimension. But, for ease of exposition, we work with scatter dimension in this section.} 

\noindent\textbf{\textsf{Scatter Dimension.}}
Informally, an $\varepsilon$-scattering sequence in a metric space $M=(\cli,\fac,\d)$ is a sequence of center-point pairs  $(x_1,p_1),\dots, (x_\ell,p_\ell)$, where $\ell$ is some positive integer and for $j\in[\ell], x_j \in \fac$, $p_j\in P$ such that  $\d(p_{j},x_{i}) \le 1$, for all $1 \le j <i\le [\ell]$ and $\d(p_{i},x_i) > (1+\varepsilon) $, for all $i \in [\ell]$.     The \emph{$\varepsilon$-scatter dimension} of $M$ is the maximum length of $\varepsilon$-scattering sequence contained in $M$. 

Now, consider a natural extension of the framework~\cite{AbbasiClustering23} for \probname in a metric space $M$ as follows. Let $\OPT_r$ be the optimal cost of the \probname instance corresponding to an optimal solution $O$. The algorithm maintains cluster constraint $Q_i$ for each cluster $i \in [k]$. Each $Q_i$ consists of a collection of requests of the form $(p,\delta)$, where $p$ is a point and  $\delta$ is a distance, representing the demand that $p$ requires a center within distance $\delta$. The  algorithm always maintains a solution $X$ such that $x_i \in X$ satisfy $Q_i$ cluster constraint, for all $i\in[k]$. 
Now consider the sequence of triples $S_i=(x^{(1)}_i,p_i^{(1)},\delta_i^{(1)}),\dots,(x^{(\ell)}_i,p_i^{(\ell)},\delta_i^{(\ell)})$ corresponding to requests in $Q_i$, where $x^{(j)}_i$ is the $i^{\text{th}}$ center maintained by the algorithm just before adding request $(p_i^{(j)},\delta_i^{(j)})$ to $Q_i$.
If $X$ is a near-optimal solution, then the algorithm terminates successfully and returns $X$. Otherwise, it identifies a point $p\in P$ whose  $r$-distance to $X$ is much larger than its $r$-distance to $O$. Such a point is called a \emph{witness} to $X$. A simple averaging argument shows that such a witness can be sampled with high probability. The algorithm then guesses the optimal cluster $i \in [k]$ of $p$ and add a new request $(p,\delta_p=\d(p,X)/(1+\varepsilon'))$ to  $Q_i$, for some suitable but fixed $\varepsilon'$ depending on $\varepsilon$. The center $x_i$  is recomputed to satisfy the updated cluster constraint $Q_i$, if possible. Otherwise the algorithm reports failure.\footnote{E.g., if the algorithm failed to sample a witness.}
The key observation is that the sequence of pairs  $(x^{(j_1)}_i,p_i^{(j_1)}),\dots,(x^{(j_\ell)}_i,p_i^{(j_\ell)})$  for a fixed radius $\delta_i^{(j)}$ forms an $\varepsilon$-scattering sequence in $M$. Thus, if the $\varepsilon$-scatter dimension of $M$ is bounded then the length of these sequences are also bounded. Furthermore, it can be shown that if the aspect ratio of the radii of requests in $Q_i$ for all  $i \in [k]$ is bounded, then the number of iterations of the algorithm can be bounded, yielding an FPT-AS.

\noindent\textbf{\textsf{Working with the inflated radius.}}
A major challenge that arises when the witness is defined in terms of $r$-distances, but the requests added to the cluster constraints are based on the true distances. Specifically, we need to ensure that a witness whose $r$-distance to $X$ is significantly larger than its $r$-distance to $O$ also has a larger true distance to $X$ than  to $O$. This  ensures that the request $(p,\d(p,X)/(1+\varepsilon'))$ is feasible and the algorithm does not fail. However, this condition may not hold, especially when the true distances are close to $r$. In fact, the issue is related to a fundamental barrier identified  by~\cite{FominGISZ24}, assuming standard conjectures.\footnote{Such a framework could potentially yield a near-optimal solution that does not violate the radius, contradicting Proposition 1 in~\cite{FominGISZ24}.}
To overcome this fundamental barrier and maintain a sufficient gap between the true distances, we consider the cost of $X$ with respect to $(1+\varepsilon)r$ radius. In other words, we look for a solution $X$ whose $(1+\varepsilon)r$-cost is close to $\OPT_r$.
In this case, we redefine a witness as a point in $P$ whose $(1+\varepsilon)r$-distance to $X$ is much larger than its $r$-distance to $O$. 
These insights allow us to establish a gap between  the true distance of a witness to $X$ and its true distance to $O$, thereby justifying the requests  added by the algorithm.

\noindent\textbf{\textsf{Bounding the aspect ratio of the radii.}}
Abbasi et al.~\cite{AbbasiClustering23} bound the aspect ratio of the radii by $(i)$ initializing the solution using ``good'' upper bounds, and $(ii)$ 
sampling witness from the points that have comparable distance to $X$ with respect to the corresponding upper bounds. The first step guarantees that the solution $X$ maintained by the algorithm always satisfies the upper bounds within constant factor. This together with  the second step allows one to bound the aspect ratio of radii in each $Q_i$. They show that such witnesses have a good probability mass if the points are sampled proportional to their true distances, and hence the algorithm successfully finds a witness for $X$ with probability that is a function of $k$ and $\varepsilon$.
Although,  devising feasible and ``good'' upper bounds for \probname can be done with some efforts, the main challenge arises in the second step, which only guarantees to sample a point whose $(1+\varepsilon)r$-distance to $X$ is comparable with its upper bound. As mentioned before, these $(1+\varepsilon)r$-distances can be much smaller than the corresponding true distances, and hence there is no guarantee that true distance of a witness $p$ to current solution $X$ is comparable to  its upper bound. Towards this, we introduce a novel idea of dividing witnesses into two sets --- ``nearby witness'' set, i.e., the set of points within distance $O(r/\varepsilon)$ from $X$, and ``faraway witness'' set, which are beyond  this distance from $X$.
The observation is that, since the total probability mass on the witness set, now defined using $(1+\varepsilon)r$-distances, is still high, it must be that either the nearby or the faraway witnesses have sufficient probability mass. However, since we do not know which of the two sets has enough mass, we perform a randomized branching (i.e., make a guess)---which will be ``correct'' with probability $1/2$. Then, when a point is sampled proportional to its $(1+\varepsilon)r$-distance from either the ``nearby'' or ``faraway'' set, it will be a witness with good probability. Now consider each of two cases separately. Since for a nearby witness $p$, its true distance to $X$ is at least $(1+\varepsilon)r$  and at most $O(r/\varepsilon)$, the aspect ratio of the radii of requests corresponding to nearby witness set is bounded by $O(1/\varepsilon)$. On the other hand, for requests corresponding to faraway witness set, we show that their radii lie in bounded interval, using ideas similar to~\cite{AbbasiClustering23}. Note that, these two arguments imply that the radii of the requests lie in two (possibly disjoint)  intervals that themselves are bounded. However, it is not clear if the length of these requests is bounded, unlike~\cite{AbbasiClustering23}, where the radii belonged to a single interval of bounded aspect ratio. Nevertheless, we observe that, using the techniques of~\cite{AbbasiClustering23}, the length of requests can be bounded, even when the radii lie in constantly many (here, two)  intervals, each with a bounded aspect ratio.

\section{Background on Algorithmic Scatter Dimension} \label{sec:prelims}

In this paper, we consider metric (clustering) space $M=(\cli,\fac,\d)$, where $\cli$ is a finite set of $n$ points, $\fac$ is a (possible infinite) set of potential cluster centers, and $\d$ is a metric on $(\cli \cup \fac)$. A class $\mathcal{M}$ of metric spaces is a (possibly infinite) set of metric spaces.

In this paper, we work with a notion that is weaker (and hence more general) than $\varepsilon$-scatter dimension that was defined in the overview, called \emph{algorithmic $\varepsilon$-scatter dimension}, which we explain next. To this end, we first need the following search problem.
\begin{definition}[\ballint Problem]\label{def:ball-int}
	Let $\mathcal{M}$ be a class of metric spaces with possibly infinite set of centers. Given $M=(\cli,\fac,\d)\in \mathcal{M}$, a finite set $Q \subsetneq P \times \mathbb{R}_+$ of distance constraints, and an error parameter $\eta >0$, the \emph{Ball Intersection} problem asks to find a center $x \in F$ that satisfy all distance constraints within $\eta$ multiplicative error, i.e., $\d(x,p) \le (1+\eta)\delta$, for every $(p,\delta)\in Q$, if such a center exists, and report failure otherwise.
	
	We say $\mathcal{M}$ admits a \ballint algorithm  if it correctly solves the ball intersection problem for every metric space in $M$ and runs in polynomial time in the size of $M$ and $1/\eta$.
\end{definition}
Now, we are ready to define algorithmic scatter dimension.
\begin{definition}[Algorithmic $\varepsilon$-Scatter Dimension]\label{def:algsd}
	Given a class $\mathcal{M}$ of metric spaces with \ballint algorithm $\mathcal{C}_\mathcal{M}$, a space $M \in \mathcal{M}$, and $\varepsilon \in (0,1)$,  a \emph{$(\mathcal{C}_\mathcal{M},\varepsilon)$-scattering sequence} is a sequence $(x_1,p_1,\delta_1),\dots,(x_\ell,p_\ell,\delta_\ell)$, where $\ell$ is some positive integer, and for $i \in [\ell], x_i \in \fac$, $p_i \in \cli$ and $\delta_i \in \mathbb{R}_+$ such that
	\begin{itemize}
		\item (Covering by $\mathcal{C}_\mathcal{M}$)\quad $x_i = \mathcal{C}_\mathcal{M}(M,\{(p_1,\delta_1),\dots,(p_{i-1},\delta_{i-1})\},\varepsilon/2)$ \quad $\forall 2 \le i \le \ell$
		\item ($\varepsilon$-refutation)\quad\quad\quad $\d(x_i,p_i) > (1+\varepsilon)\delta_i$ \qquad $\forall i \in [\ell]$
	\end{itemize}
	The \emph{algorithmic $(\varepsilon,\mathcal{C}_\mathcal{M})$-scatter dimension} of $\mathcal{M}$ is $\lambda_{\mathcal{M}}(\varepsilon)$ if any $(\mathcal{C}_\mathcal{M},\varepsilon)$-scattering sequence contains at most $\lambda_{\mathcal{M}}(\varepsilon)$ many triples per radius value. The \emph{algorithmic $\varepsilon$-scatter dimension} of $\mathcal{M}$ is the minimum algorithmic $(\varepsilon,\mathcal{C}_\mathcal{M})$-scatter dimension over any \ballint algorithm $\mathcal{C}_\mathcal{M}$ for $\mathcal{M}$.
\end{definition}

Although, algorithmic scatter dimension restricts the number of triples in the sequence with same radius value, we can use the proof technique from~\cite{AbbasiClustering23} to prove the following stronger guarantee. 
\begin{restatable}{lemma}{scdoint}\label{cor:scdonint} 
	Let $\mathcal{M}$ be a class of metric spaces of algorithmic $\varepsilon$-scatter dimension $\lambda(\varepsilon)$. Then there
	exists a \ballint\ algorithm $\mathcal{C}_{\mathcal{M}}$ with the following property. Given $\varepsilon \in (0,1)$, a constant $t\ge 1$, and $a_i >0, \tau_i\ge 2$ for $i \in [t]$, any $(\mathcal{C}_{\mathcal{M}}, \varepsilon)$-scattering contains $O(\sum_{i\in [t]}\lambda(\nicefrac{\varepsilon}{2}) (\log\tau_{i})/\varepsilon)$ many triples whose radii lie in the interval $\cup_{i \in [t]}[a_i, \tau_i a_i]$.
\end{restatable}
\iffull
\begin{proof}
	We first need the following result from \cite{AbbasiClustering23}.
	
	\begin{proposition}[\cite{AbbasiClustering23}]\label{lem:scdonint} 
		Let $\mathcal{M}$ be a class of metric spaces of algorithmic $\varepsilon$-scatter dimension $\lambda(\varepsilon)$. Then there
		exists a \ballint\ algorithm $\mathcal{C}_{\mathcal{M}}$ with the following property. Given $\varepsilon \in (0,1), a >0$
		and $\tau\ge 2$, any $(\mathcal{C}_{\mathcal{M}}, \varepsilon)$-scattering contains $O(\lambda(\nicefrac{\varepsilon}{2}) (\log\tau)/\varepsilon)$ many triples whose radii lie in the interval $[a, \tau a]$.
	\end{proposition}
	
	Now, by \Cref{lem:scdonint}, we have that there exists  $\ballintalg$\ algorithm $\mathcal{C}_{\mathcal{M}}$
	such that any $(\mathcal{C}_{\mathcal{M}}, \varepsilon)$-scattering contains $O(\lambda(\nicefrac{\varepsilon}{2}) (\log\tau_{i})/\varepsilon)$ many triples whose radii lie in the interval $[a_i, \tau_i a_i]$,  for $i \in [t]$. Therefore, any $(\mathcal{C}_{\mathcal{M}}, \varepsilon)$-scattering contains $\sum_{i\in[t]}O(\lambda(\nicefrac{\varepsilon}{2}) (\log\tau_{i})/\varepsilon)$ many triples whose radii lie in the interval $\cup_{i\in[t]}[a_i, \tau_i a_i]$.
\end{proof}
\fi

\section{Bicriteria FPT Approximation Scheme} \label{sec:algorithm}

\subsection{Algorithm} \label{subsec:alg}
\begin{algorithm}[h]
	\caption{Approximation Scheme for \probname}
	\label{algo:apx}
	\begin{algorithmic}[1]
		\Statex \textbf{Input:} Instance $\cI= ((\cli,\fac, \d), k, r)$ of \probname, $\varepsilon \in (0, 1)$, and \textsc{Ball Intersection} algorithm \ballintalg, and a guess $\mathcal{G}$ for $\OPT_r$
		\Statex \textbf{Output:} A solution $X \subseteq \fac$ such that $\cost_{(1+\varepsilon)r}(\cli, X) \le (1+\varepsilon) \mathcal{G}$, assuming $\OPT_r \le \mathcal{G} \le (1+\varepsilon/3) \cdot \OPT_r$. 
		\State For each $p \in \cli$, compute $u(p)= 3\cdot \min\{\alpha>r : |\ball(p,\alpha)| \ge \mathcal{G}/\alpha \}$\label[lin]{lin:ub}
		\State Process $P$ in non-decreasing order of $u(p)$ and mark $p_i \in P$ if $\ball(p_i,u(p_i))$ is disjoint from $\ball(p_j,u(p_j))$ for every marked $p_j$ such that $j <i$ \label[lin]{lin:mis}
		\State Let $p^{(1)}, \dots,p^{(k')}$ be the marked points \label[lin]{lin:initpoints}
		\State \textbf{for each} $i \in [k']$, let $Q_i \coloneqq \LR{(p^{(i)}, u(p^{(i)}))}$\label{lin:ubs} 
		\State For $k\ge i >k'$, let $Q_i =  \emptyset$ \label{lin:emptyub}
		\State Let $X \coloneqq (x_1, \ldots, x_k)$, where $\forall i \in [k]$, $x_i \in \fac$ is any center satisfying requests in $Q_i$. \label[lin]{lin:initsolution}
		\State Let $r' \coloneqq r(1+\sfrac{\varepsilon}{3})$. \label{lin:rprime}
		\While{$\cost_{r'}(P, X) > (1+\varepsilon) \cdot \mathcal{G}$} \label{lin:whilestart}
		
		\State Toss a fair coin to guess whether we are in ``nearby witness'' or ``faraway witness'' case
		\If{we guess ``nearby witness'' case} \label[lin]{lin:cointoss} \label[lin]{lin:if}
		\State $N = \LR{p \in P: \d(p, X) \le \frac{8r}{\varepsilon} }$ \label{lin:setN}
		\State Sample a point $p \in N$, where $\Pr(p = a) = \frac{\d_{r'}(a, X_i)}{\sum_{b \in N} \d_{r'}(b, X_i)}$ for each $a \in \cli$ \label[lin]{lin:sample1}
		\Else{ we guess ``faraway witness case \textbf{then}} \label[lin]{lin:else}
		\State Let $A \coloneqq \LR{ p \in P : \d_{r'}(p, X) > \frac{\varepsilon}{1000k} \cdot u(p) }$\label[lin]{lin:setA}
		\State Sample a point $p \in A$, where $\Pr(p = a) = \frac{\d_{r'}(a, X_i)}{\sum_{b \in \cli} \d_{r'}(b, X_i)}$ for each $a \in A$ \label[lin]{lin:sample2}
		\EndIf
		\State Sample an integer $i \in [k]$ u.a.r. \label[lin]{lin:clusterindex}
		\State Add $(p, \delta_p)$ to $Q_i$, where $\delta_p = \frac{\d(p, X)}{1+\sfrac{\varepsilon}{12}}$ \label{lin:add request} 
		\State $x_i \gets \ballintalg (Q_i, \fac,\sfrac{\varepsilon}{40})$ \textbf{if} no $x_i$ was found then \textbf{fail} \label[lin]{lin:ballalgo}
		\EndWhile \label{lin:whileend}
	\end{algorithmic}
\end{algorithm}

Our bicrteria \fptas for \probname is formally stated in \Cref{algo:apx}. As an input, we are given an instance $\mathcal{I} = ((\cli, \fac, \d), k, r)$ of \probname, an accuracy parameter $\varepsilon$, access to an algorithm \ballintalg for the so-called ``{\sc Ball Intersection}'' problem (discussed later), and a guess $\mathcal{G}$ for the optimal cost $\OPT_r$. By a standard exponential search, we will assume that $\OPT_r \le \mathcal{G} \le (1+\varepsilon/3) \cdot \OPT_r$. At a high level, this algorithm can be divided into two steps: \emph{initialization phase} and the \emph{iterative cost-improvement phase}. The initialization phase spans from line \ref{lin:ub} to \ref{lin:rprime}. 

At a high level, the goal of this phase to compute for each point $p \in P$, an \emph{upper bound} $u(p)$, such that $p$ must have an optimal center within distance $u(p)$. Once we find such upper bounds, we use a subset of them to initialize for each $i \in [k]$, a set of \emph{requests} $Q_i$, where each request $(p, \delta_p)$ demands that the $i$th center in the solution must be within distance at most $\delta_p$ from $p$ in every subsequent solution found by the algorithm.

Now, the algorithm moves to the iterative cost-improvement phase in lines \ref{lin:whilestart} to \ref{lin:whileend}, consisting of a while loop that runs as long as the current solution has not become the bicriteria approximation that we are looking for. Thus, in each iteration of the while loop, we are given a solution $X$ that satisfies all the requests $Q_i$, and yet satisfies $\cost_{r'}(P, X) > (1+\varepsilon) \cdot \mathcal{G}$. Then, our algorithm makes a random choice whether there is enough cost-contribution of nearby witnesses, or of faraway witnesses---here a witness is a point $p$ whose distance to $X$ is sufficiently larger than that in the (unknown) optimal solution $O$. In each of the cases, we sample points from carefully defined sets (cf. $N$ in line \ref{lin:setN} and $A$ in \ref{lin:setA}), proportional to their contribution to the cost to the respective sets (in the analysis, we will argue that with good probability, we will in fact sample witness points). Having sampled such a point $p$, we guess the index $i$ of the optimal cluster of $p$ in line \ref{lin:clusterindex}. Finally, assuming $p$ is indeed a witness, we add a request $(p, \frac{\d(p, X)}{1+\varepsilon/12})$ to the $i$th request set $Q_i$, and in line \ref{lin:ballalgo} we recompute $x_i$ using \ballintalg. This algorithm either returns a center $x_i \in \fac$ that satisfies all the requests (with an error of up to a factor of $\varepsilon/40$), or correctly outputs that there is no point in $\fac$ satisfying all the requests. Note that in discrete metric spaces, \ballintalg can be simulated by simply scanning each $x \in \fac$ and checking whether it satisfies all the requests in $Q_i$. On the other hand, for continuous Euclidean spaces, such an algorithm was designed in \cite{AbbasiClustering23}. If the algorithm does not fail at this step, then we update our set $X$ and continue to the next iteration. In the rest of the section, we will prove that the algorithm returns a $(1+\varepsilon, 1+\varepsilon)$-bicriteria approximation with good probability.

\subsection{Analysis}
Throughout the analysis, we assume that we are given a guess $\mathcal{G}$, such that $\OPT_r \le \mathcal{G} \le (1+\varepsilon/3)\cdot \OPT_r$. We divide the analysis in two parts. In the first part, we bound the running time of the algorithm using the following lemma. The proof of this lemma is present in Section~\ref{ss:runtime}.
\begin{lemma}\label{lem:runtime}
	Algorithm~\ref{algo:apx} terminates in $O(\frac{k}{\varepsilon}\log (\frac{k}{\varepsilon})\lambda(\frac{\varepsilon}{40}))$ iterations---with or without failure.
\end{lemma}

In the second part, we show the following lemma, which says that the probability that the algorithm terminates without failure is high. The proof the lemma is present in Section~\ref{ss:successprob}.
\begin{lemma}\label{lem:successprob}
	With probability at least \sucprob, Algorithm~\ref{algo:apx} terminates without failure, i.e., returns a solution $X$ satisfying $\cost_{(1+\varepsilon/3)r}(P, X) \le (1+\varepsilon) \OPT_r$.
\end{lemma}

Using these two lemmas and repeating the algorithm $\exp\lr{O\lr{\frac{k}{\varepsilon} \log(\frac{k}{\varepsilon}) \lambda(\frac{\varepsilon}{40}) }}$ times, we have our main result.
\begin{theorem}[Main Theorem]\label{thm:main}
	Let $\mathcal{M}$ be a class of metric spaces closed under scaling distances by a positive constant. There is a randomized algorithm that takes as input an instance  $\cI=((\cli, \fac, \d), k, r)$ of \probname such that $(\cli,\fac,\d) \in \mathcal{M}$ and  $\varepsilon \in (0,1)$, and outputs a $(1+\varepsilon,1+\varepsilon)$-bicriteria solution for $\cI$ when, for all $\varepsilon'>0$, the algorithmic $\varepsilon'$-scatter dimension of $\mathcal{M}$ is bounded by $\lambda(\varepsilon')$, for some function $\lambda$. The running time of the algorithm is $2^{O(\frac{k}{\varepsilon}\cdot\log (k/\varepsilon)\cdot \lambda(\varepsilon/40))} \cdot |{\cI}|^{O(1)}$.
	
\end{theorem}
\subsubsection{Bounding runtime using Algorithmic Scatter Dimension}\label{ss:runtime}

First, we show some properties of the initial upper bounds (line~\ref{lin:ubs}), that we need later in the proof of Lemma~\ref{lem:runtime}.
\begin{lemma}[Feasible upper bounds]\label{lem:upperboundfeas}
	Consider $Q_i = \{(p^{(i)},u(p^{(i)}))\}$ initialized in Line~\ref{lin:ubs} of Algorithm~\ref{algo:apx}. Then, $\d(p^{(i)},O) \le u(p^{(i)})$. 
\end{lemma}
\begin{proof}
	Suppose $\d(p^{(i)},O) > u(p^{(i)})$. Letting $\alpha = u(p^{(i)})/3$, we have that $|\ball(p,\alpha)| \ge \mathcal{G}/\alpha$. Since,  $\d(p^{(i)},O) > 3\alpha$, we have $\d(p,O) > 2\alpha$ for $p \in  \ball(p^{(i)},\alpha)$. Using $\alpha>r$, this means $\d_r(p,O) >\alpha$  for $p \in  \ball(p,\alpha)$. Therefore,
	\[
	\cost_r(P,O) \ge \sum_{p \in \ball(p^{(i)},\alpha)}\d_r(p',O) =  \sum_{p \in \ball(p^{(i)},\alpha)}(\d(p,O) - r)> \frac{\mathcal{G}}{\alpha}\cdot \alpha =\mathcal{G},
	\]
	contradicting the the cost of $O$. 
\end{proof}

Next, we have the following lemma, whose proof is identical to~\cite{AbbasiClustering23}, that says that the initialization of $X$ at line~\ref{lin:initsolution} is successful and the solution maintained by the algorithm always satisfies the upper bounds within a factor of $3.1$. 
\begin{lemma}[Lemma V.5 of ~\cite{AbbasiClustering23}]\label{lem:3apxsol}
	The number of marked points in line~\ref{lin:initpoints} is at most $k$, i.e., $k' \le k$. Hence, the initialization of $X$ at line~\ref{lin:initsolution} is successful.
	Furthermore, at any iteration, the solution $X$  maintained by the algorithm satisfies that $\d(p,X) \le 3.1u(p)$, for every $p \in P$. 
\end{lemma}
\paragraph*{Bounding aspect ratio of radii}
Towards proving Lemma~\ref{lem:runtime}, we bound the aspect ratio of the radii in the requests.
\begin{lemma}\label{lem:scseq bounded}
	Consider a request set $Q_i=\{(p^{(1)}_i,\delta^{(1)}_i),\dots,(p^{(\ell)}_i,\delta^{(\ell)}_i)\}$, for $i \in [k]$.
	Let $X^{(j)}, j \in [\ell]$,  be the center maintained by the algorithm just before adding the request $(p^{(j)}_i,\delta^{(j)}_i)$ to $Q_i$. Further, let $x^{(j)}_i \in X^{(j)}$  be the center corresponding to cluster $i \in [k]$.
	Then, the sequence    $S_i =(x^{(1)}_i,p^{(1)}_i,\delta^{(1)}_i),\dots,(x^{(\ell)}_i,p^{(\ell)}_i,\delta^{(\ell)}_i)$ 
	is an algorithmic $(\ballintalg,\nicefrac{\varepsilon}{20})$-scattering. Furthermore, the  radii of requests in $S_i$  lie in the interval $[r,8r/\varepsilon] \cup [r_{min},\frac{10^5k}{\varepsilon^2}r_{min}]$, where $r_{min}$ is the smallest radii in $S_i$ that is larger than $8r/\varepsilon$.
\end{lemma}
\begin{proof}
	The proof of the first part is similar to~\cite{AbbasiClustering23}.
	
	First note that $\delta^{(j)}_i = \frac{\d(p^{(j)}_i,X^{(j)})}{1+\nicefrac{\varepsilon}{12}} \le \frac{\d(p^{(j)}_i,x^{(j)}_i)}{1+\nicefrac{\varepsilon}{12}}$. Finally, since $x^{(j)}_i$ is computed by \ballintalg\ on $\{(p^{(1)}_i,\delta^{(1)}_i),\dots,(p^{(j-1)}_i,\delta^{(j-1)}_{i})\}$ and error parameter $\varepsilon/40$, we have that $\d(p^{(j')}_i,x^{(j)}_i) \le (1+\nicefrac{\varepsilon}{40})\delta^{(j')}_i$, for $j' <j$. Hence, $S_i$ is $(\ballintalg,\nicefrac{\varepsilon}{20})$-algorithmic scattering. 
	
	Now we bound that the radii of sequence $S_i, i \in[k]$. Towards this, we partition $S_i$ into $S^n_i$ and $S^f_i$, based on a partitioning of $Q_i$, as follows. 
	We say a request $(p^{(j)}_i,\delta^{(j)}_i) \in Q_i$ belongs to $Q^n_i$ if $p^{(j)}_i$ was sampled when the If condition was satisfied (in Line~\ref{lin:sample1}), otherwise it belongs to $Q_{i}^{f}$, which  corresponds to the Else condition (in Line~\ref{lin:sample2}). 
	Correspondingly,  $(x^{(j)}_i,p^{(j)}_i,\delta^{(j)}_i) \in S^n_i$  if $(p^{(j)}_i,\delta^{(j)}_i) \in Q^n_i$ and $(x^{(j)}_i,p^{(j)}_i,\delta^{(j)}_i) \in S^f_i$  if $(p^{(j)}_i,\delta^{(j)}_i) \in Q^f_i$.
	We bound the aspect ratio of each part, $S^n_i$ and $S^f_i$, separately. For $(x^{(j)}_i,p^{(j)}_i,\delta^{(j)}_i)\in S_{i}^{n}$, we have $r\le\delta^{(j)}_i\le8r/\varepsilon$ and hence, the radii of triples in $S^n_i$ lie in interval $[r,8r/\varepsilon].$
	Now consider $(x^{(j)}_i,p^{(j)}_i,\delta^{(j)}_i)\in S_{i}^{f}$. Recall that $X^{(j)}$ is the solution maintained by the algorithm when $(p^{(j)}_i,\delta^{(j)}_i)$ is added to $Q_{i}$. Then, note that $\d_{r'}(p^{(j)}_i,X^{(j)})>\varepsilon u(p)/1000k$ (see Line 14),
	and hence, we have $\d(p^{(j)}_i,X^{(j)})>\varepsilon u(p)/1000k$.
	The following claim, whose proof is similar to~\cite{AbbasiClustering23}, shows that the radii of requests in $S^f_i$ are also bounded, finishing the proof of the lemma. 
	\begin{claim}\label{cl:bndscseq}
		Consider        requests $(p,\delta_p),(p',\delta_{p'})$  added (in any order) to $Q_i$  in Line~\ref{lin:add request} of Algorithm~\ref{algo:apx}  such that $(p,\delta_p),(p',\delta_{p'}) \in Q^f_i$. If $\delta_{p'} < \nicefrac{\varepsilon^2 \delta_p}{10^5k}$, then the algorithm fails in Line~\ref{lin:ballalgo} upon making second of the two requests.
	\end{claim}
	\begin{claimproof}
		Suppose, for the sake of contradiction, the algorithm does not fail and finds a center $x_i$ such that $\d_(p,x_i) \le (1+\varepsilon/40)\delta_p$ and $\d(p',x_i)\le (1+\varepsilon/40)\delta_{p'}$. Thus, $\d(p,p') \le \d(p,x_i) + \d(p',x_i) \le (1+\varepsilon/40)(\delta_p + \delta_{p'})$.
		
		Therefore, we have
		\[
		\d(p,X) \le \d(p,p') + \d(p',X) \le  (1+\varepsilon/40)(\delta_p + \delta_{p'}) + 3.1u(p'),
		\]
		using Observation~\ref{lem:3apxsol}.
		Let $X$ be the center maintained by the algorithm when the request $(p,\delta_p)$ is added to $Q_i$. Then, we have
		\begin{equation}
			\delta_p = \frac{\d(p,X)}{(1+\varepsilon/12)},\label{eqn:rad_p}
		\end{equation}
		due to Line~\ref{lin:add request}. 
		Similarly, let $X'$ be the center maintained by the algorithm when request $(p',\delta_{p'})$ is added to $Q_i$. Then, since $(p',\delta_{p'}) \in Q^f_i$, we have that $u(p') \le 900 \d(p',X')/\varepsilon = 1000k\delta_{p'}/\varepsilon$, using $\delta_{p'} = \frac{\d(p',X')}{(1+\varepsilon/12)}$. Therefore, using  $\delta_{p'} < \nicefrac{\varepsilon^2 \delta_p}{10^5k}$,
		\begin{align*}
			\d(p,X) &\le (1+\nicefrac{\varepsilon}{40})(\delta_p + \delta_{p'}) + {3000k\delta_{p'}}/{\varepsilon} \le (1+\nicefrac{\varepsilon}{40})\delta_p + {3200k\delta_{p'}}/{\varepsilon} \\
			&\le (1+\nicefrac{\varepsilon}{40} + \nicefrac{\varepsilon}{25})\delta_p< (1+\nicefrac{\varepsilon}{12})\delta_p
		\end{align*}
		This means $\delta_p > \frac{\d(p,X)}{1+\varepsilon/12}$, contradicting (\ref{eqn:rad_p}).
	\end{claimproof}
\end{proof}
Now we are ready to finish the proof of \Cref{lem:runtime}.
\begin{proof}[Proof of Lemma~\ref{lem:runtime}]
	We apply Corollary~\ref{cor:scdonint} to  $S_i$ and note that the radii in $S_i$ lie in $[r,8r/\varepsilon] \cup [r_{min},\frac{10^5k}{\varepsilon^2}r_{min}]$, where $r_{min}$ is the smallest radii in $S_i$ that is larger than $8r/\varepsilon$, due to  Lemma~\ref{lem:scseq bounded}. This implies that the length of sequence $S_i$ is bounded by $O(\lambda(\nicefrac{\varepsilon}{40}) \frac{(\log (\nicefrac{k}{\varepsilon}))}{\varepsilon})$. Since in each iteration the algorithm adds one request to some $Q_i$, the total number of iterations is bounded by $O(\frac{k}{\varepsilon}\lambda(\nicefrac{\varepsilon}{40})(\log (\nicefrac{k}{\varepsilon}))$.
\end{proof}

\subsubsection{Bounding success probability}\label{ss:successprob}

The goal of this subsection is to prove \Cref{lem:successprob}, i.e., give a lower bound on the success probability of the algorithm. To this end, we first introduce the following notion.

\begin{definition}[Consistency]
	Consider a fixed hypothetical optimal solution $O=(o_{1},\cdots,o_{k})$.
	We say that the current state of execution (specified by $(X,Q_{1},\cdots,Q_{k})$)
	of Algorithm 1 is \emph{consistent} with $O$ if for any request $(p,\delta)\in Q_{i},i\in[k]$,
	we have $\d(p,o_{i})\le\delta$.
\end{definition}

Note that, \Cref{lem:3apxsol} implies that the initial set of requests (line \ref{lin:ubs}) are feasible. Since the balls $\ball(p^{(i)},u(p^{(i)})), i \in [k']$ are disjoint, we can relabel the optimum centers $O=\{o_1,\dots,o_k\}$ so that $\d(p^{(i)},o_i)\le u(p^{(i)})$. Thus, \Cref{lem:3apxsol} implies that the initial state of the algorithm as computed in lines \ref{lin:ub} to \ref{lin:initsolution} is consistent with a fixed optimal solution $O$ which is fixed henceforth. To prove \Cref{lem:successprob}, we will inductively argue that the state of the algorithm remains consistent with $O$ -- note that the base case is already shown. Most of this subsection is devoted to show the inductive step, i.e., to show that, given a consistent state at the start of an iteration, there is a good probability that the state of the algorithm at the end of the iteration is also consistent (cf. \Cref{lem:consistency}).

To this end, we introduce the following definitions w.r.t.~the current solution $X$.
\begin{definition}[Contribution and Witness]
	For any set $S \subseteq \cli$, we use $C_S \coloneqq \sum_{p \in S} \d_{r'}(p, X)$ to denote the \emph{contribution} of the set $S$ to the cost of the current solution.
	
	We say that a point $p \in P$ is an \emph{$\varepsilon/3$-witness} (or simply, a \emph{witness}) w.r.t.~the solution $X$ if it satisfies $\d_{(1+\sfrac{\varepsilon}{3})r}(p, X) > (1+\sfrac{\varepsilon}{3}) \cdot \d_r(p, O)$.
\end{definition}
The following claim can be proved by a simple averaging argument.
\begin{restatable}{claim}{wcontri} \restatemarker \label{cl:wcontri} 
	$C_W \ge \frac{\varepsilon C_\cli}{10}$.
\end{restatable}
\iffull
\begin{claimproof}
	Suppose not. Then,
	\begin{align*}
		\OPT_r &\ge \sum_{p \in \cli \setminus W} \d_r(p, O) 
		\\&\ge \frac{1}{1+\sfrac{\varepsilon}{3}} \sum_{p \in \cli \setminus W} \d_{r'}(p, X) \tag{$\forall p \in \cli \setminus W$, $\d_{r'}(p, X) \le (1+\sfrac{\varepsilon}{3}) \cdot \d_r(p, O)$}
		\\&\ge \frac{1-\sfrac{\varepsilon}{10}}{1+\sfrac{\varepsilon}{3}} \sum_{p \in \cli} \d_{r'}(p, X) \tag{By assumption, $C_{W} < \frac{\varepsilon C_P}{10}$}
		\\&\ge \frac{1}{1+\sfrac{\varepsilon}{2}} \cdot \cost_{r'}(\cli, X)
	\end{align*}
	The last inequality can be rearranged to $\cost_{r'}(\cli, X) \le (1+\sfrac{\varepsilon}{2}) \cdot \OPT_r \le (1+\sfrac{\varepsilon}{2})(1+\sfrac{\varepsilon}{3}) \cdot \mathcal{G} \le (1+\varepsilon) \mathcal{G}$, contradicting that $\cost_{r'}(\cli, X) > (1+\varepsilon) \cdot \mathcal{G}$.
\end{claimproof}
\fi
Next, we introduce several different classifications of witnesses. 
\begin{definition}[Different subsets of witnesses] \label{def:witnesstypes}
	For each $x_j \in X$, let $W_j$ denote the set of witnesses for which $x_j$ is a nearest center in $X$ (breaking ties arbitrarily). Then, $W_{j, {\sf near}} \coloneqq \LR{p \in W_j : \d(p, x_j) \le \sfrac{8r}{\varepsilon}}$, and $W_{j, {\sf far}} \coloneqq W_j \setminus W_{j, {\sf near}}$. Further, let $W_{\sf near} \coloneqq \bigcup_{j \in [k]} W_{j, {\sf near}}$, and $W_{\sf far} \coloneqq \bigcup_{j \in [k]} W_{j, {\sf far}}$. We will refer to a witness in $W_{\sf near}$ as a \emph{nearby witness} and a witness in $W_{\sf far}$ as a \emph{faraway witness}.
\end{definition}

Now, we consider two different cases regarding the behavior of the algorithm.

\subparagraph{Case 1: $C_{W_{\sf near}} \ge \frac{\varepsilon}{100} C_P$.} In this case, when we sample a point from $N \coloneqq \{p \in P: \d(p, X) \le 8r/\varepsilon\}$ proportional to their $\d_{r'}(\cdot, X)$ values, the probability of sampling a nearby witness is at $\frac{\varepsilon}{100}$. This will correspond to the ``good event''. We prove this formally in the following lemma.

\begin{lemma}[Nearby witness lemma] \label{lem:nearbywitnesslemma}
	Suppose the current solution $X$ at the start of an iteration satisfies $\cost_{r'}(P, X) > (1+\varepsilon) \cdot \mathcal{G}$. Further, suppose $C_{W_{\sf near}} \ge \frac{\varepsilon}{100} C_P$. Then, with probability at least $\frac{\varepsilon}{200k}$, the point $p \in P$, the index $i \in [k]$, and value $\delta_p$ defined in the iteration satisfy the following properties.
	\begin{enumerate}
		\item $o_i \in O$ is the closest center to $p \in P$, i.e., $\d(p, o_i) = \d(p, O)$, 
		\item $p \in W_{\sf near}$, i.e., (i) $\d_{r'(p, X)} > (1+\sfrac{\varepsilon}{3}) \d_{r}(p, o_i)$, and (ii) $\d(p, X) \le \frac{8r}{\varepsilon}$.
		\item $\d(p, o_i) < \frac{\d(p, X)}{1+\sfrac{\varepsilon}{12}} =: \delta_p \le \frac{8r}{\varepsilon}$. 
	\end{enumerate}
\end{lemma}
\begin{proof}
	First, in line \ref{lin:cointoss} with probability $1/2$, we correctly move to the ``nearby witness'' (if) case of line \ref{lin:if}. We condition on this event. Next, note that $W_{\sf near} \subseteq N$, and $C_{W_{\sf near}} \ge \frac{\varepsilon}{100}C_P \ge \frac{\varepsilon}{100} C_N$, where the first inequality is due to the case assumption. Therefore, in line \ref{lin:sample1}, the point $p$ sampled from $N$, will belong to $W_{\sf near}$ with probability at least $\frac{\varepsilon}{100}$. Let $i \in [k]$ denote the index such that $o_i \in O$ is the closest center in the optimal solution $O$, and we correctly guess the index $i$ in line \ref{lin:clusterindex}. Note that the probability of the algorithm making the ``correct'' random choices is at least $\frac{1}{2} \cdot \frac{\varepsilon}{100} \cdot \frac{1}{k} = \frac{\varepsilon}{200k}$.
	
	We condition on the said events. Since $p$ is a witness, we know that $\d_{r'}(p, X) > (1+\sfrac{\varepsilon}{3})\cdot \d_{r}(p, o_i)$. We first observe that $\d_{r'}(p, X)$ must be positive, which implies that $\d(p, X) > r' = (1+\sfrac{\varepsilon}{3})r$. Now, we consider two cases based on the value of $\d(p, O)$.
	
	If $\d(p, o_i) < r$, then $\d_{r}(p, o_i) = 0$, in which case, 
	\begin{equation}
		\d(p, o_i) < r \le \frac{r'}{1+\sfrac{\varepsilon}{12}} \le \frac{\d(p, X)}{1+\sfrac{\varepsilon}{12}} \label{eqn:case11}
	\end{equation}
	
	Otherwise, $\d(p, o_i) \ge r$, in which case $\d_r(p, o_i) = \d(p, o_i) - r$. Then, consider
	\begin{align*}
		&\d(p, X) - (1+\tfrac{\varepsilon}{3})r = \d_{r'}(p, X) > (1+\tfrac{\varepsilon}{3}) \cdot \d_{r}(p, o_i) = (1+\tfrac{\varepsilon}{3}) (\d(p, o_i) - r)
		\\\implies &\d(p, X) - (1+\tfrac{\varepsilon}{3})r > (1+\tfrac{\varepsilon}{3}) \cdot \d(p, o_i) - (1+\tfrac{\varepsilon}{3})r 
		\\\implies &\d(p, o_i) < \frac{\d(p, X)}{1+\sfrac{\varepsilon}{3}} \le \frac{\d(p, X)}{1+\sfrac{\varepsilon}{12}} \numberthis \label{eqn:case12}
	\end{align*}
	Thus, regardless of the value of $\d(p, o_i)$, we have established the third item, hence completing the proof of the lemma.
\end{proof}

\subparagraph{Case 2: $C_{W_{\sf near}} < \frac{\varepsilon}{100} C_P$.} In this case, most of contribution of witnesses is coming from ``faraway witnesses''. In this case, the ``correct choice'' corresponds to the case when we sample points from the set $A = \LR{ p \in P : \d_{r'}(p, X) > \frac{\varepsilon}{1000k} \cdot u(p) }$ as defined in line~\ref{lin:setA}. In this case, we will show that with good probability, the sampled point is a ``faraway witness''. Specifically, we show the following lemma.
\begin{lemma}[Faraway witness lemma] \label{lem:farawaywitnesslemma}
	Suppose the current solution $X$ at the start of an iteration satisfies $\cost_{r'}(P, X) > (1+\varepsilon) \cdot \mathcal{G}$. Further, suppose $C_{W_{\sf near}} < \frac{\varepsilon}{100} C_P$. Then, with probability at least $\frac{\varepsilon}{16k}$, the point $p \in P$, the index $i \in [k]$, and value $\delta_p$ defined in the iteration satisfy the following properties.
	\begin{enumerate}
		\item $o_i \in O$ is the closest center to $p$, i.e., $\d(p, o_i) = \d(p, O)$, 
		\item $p \in W_{\sf far}$, i.e., (i) $\d_{r'(p, X)} > (1+\sfrac{\varepsilon}{3}) \d_{r}(p, o_i)$, and (ii) $\d(p, X) > \frac{8r}{\varepsilon}$,
		\item $\d_{r'}(p, X) > \frac{\varepsilon}{1000k} u(p)$, and
		\item $\d(p, o_i) < \frac{\d(p, X)}{1+\sfrac{\varepsilon}{12}} =: \delta_p \le \frac{8r}{\varepsilon}$. 
	\end{enumerate}
\end{lemma}
\begin{proof}
	First, in line \ref{lin:cointoss} with probability $1/2$, we correctly move to the ``faraway witness'' (\textbf{else}) case of line \ref{lin:else}. We condition on this event. Now, by combining \Cref{cl:wcontri} and the case assumption we obtain, 
	\begin{equation}
		C_{W_{\sf far}} = C_{W} - C_{W_{\sf near}} \ge C_P \cdot \lr{\tfrac{\varepsilon}{10} - \tfrac{\varepsilon}{100} } = \tfrac{\varepsilon}{9} C_P. \label{eqn:case2}
	\end{equation}
	Next, let $H \coloneqq \LR{ j \in [k]: C_{W_{j, {\sf far}}} \ge \frac{\varepsilon C_P}{100k}}$. It follows that,
	\begin{align}
		\sum_{j \in [k] \setminus H}C_{W_{j, {\sf far}}} &\le k \cdot \frac{\varepsilon C_P}{100k} \nonumber
		\\\implies \sum_{j \in H} C_{W_{j, {\sf far}}} &\ge C_{W_{\sf far}} - \sum_{j \in [k] \setminus H}C_{W_{j, {\sf far}}} \ge \frac{\varepsilon C_P}{9} - \frac{\varepsilon C_P}{100} \ge  \frac{\varepsilon C_P}{8} \label{eqn:mostcontri}
	\end{align}
	Fix a $j \in H$, and let us index the points in $W_{j, {\sf far}} = \LR{z_1, z_2, \ldots, z_{\ell}}$ in the non-decreasing order of their distances $\d(z, x_j)$ (and equivalently, $\d_r(z, x_j)$). First, we state the following simple consequences of various definitions for future reference.
	\begin{restatable}{observation}{distances} \label{obs:distances}
		For each $z \in W_{j, {\sf far}}$, the following bounds hold.
		\begin{enumerate}
			\item $\d(z, x_j) \ge \frac{8r}{\varepsilon}$, 
			\item $\d(z, O) > (1+\frac{\varepsilon}{3}) \cdot \d(z, x_j)$
			\item $\d_{r'}(z, x_j) \ge \frac{6r}{\varepsilon}$, and
			\item $\d_r(z, x_j) \le \d(z, x_j) \le (1+\frac{\varepsilon}{5}) \cdot \d_{r'}(z, x_j) \le (1+\frac{\varepsilon}{5}) \cdot \d_{r}(z, x_j)$.
		\end{enumerate}
	\end{restatable}
	\iffull
	\begin{claimproof}
		Item 1 follows from the fact that $z \in W_{j, {\sf far}}$. Item 2 follows from the fact that $z$ is an $\varepsilon/3$-witness. For item 3, note that $\d_{r'}(z, x_j) = \max\LR{\d(z, x_j) - r', 0} = \d(z, x_j) - r' \ge \frac{8r}{\varepsilon} - (1+\tfrac{\varepsilon}{3}){r} > \frac{6r}{\varepsilon}$, since $\varepsilon < 1$. In item 4, the first inequality follows from the definition. For the second inequality, consider,
		\begin{align}
			\d_r(z, x_j) &= \d(z, x_j) - r \ge \d(z, x_j) - \sfrac{\varepsilon}{6} \cdot \d(z, x_j) = (1-\sfrac{\varepsilon}{6}) \cdot \d(z, x_j) \nonumber
			\\\implies \d(z, x_j) &\le \frac{1}{1-\sfrac{\varepsilon}{6}} \cdot \d_r(z, x_j) \le (1+\sfrac{\varepsilon}{5}) \cdot \d_r(z, x_j). \label{eqn:obsdistances}
		\end{align}
		This completes the proof of the observation.
	\end{claimproof}
	\fi
	Let us return to the points in $W_{j, {\sf far}} = \LR{z_1, z_2, \ldots, z_\ell}$. Let $q \in [\ell]$ denote the minimum index such that, the contribution of the set $W^{-}_{j, {\sf far}} = \LR{z_1, z_2, \ldots, z_q}$, is at least $\frac{C_{W_{j, {\sf far}}}}{2}$. Note that by a simple argument, this implies that the contribution of the set $W^{+}_{j, {\sf far}} = W_{j, {\sf far}} \setminus W^{-}_{j, {\sf far}}$ is also at least $\frac{C_{W_{j, {\sf far}}}}{3}$ (see, e.g., Lemma 7.20 in \cite{CyganFKLMPPS15}). Hence, we have the following lower bounds on the contribution of both the sets, by recalling that $j \in H$.
	\begin{equation}
		C_{W^{-}_{j, {\sf far}}} \ge \frac{\varepsilon\,C_P}{200k}, \qquad C_{W^{+}_{j, {\sf far}}} \ge \frac{\varepsilon\,C_P}{300k} \label{eqn:witnesscontribution}
	\end{equation}
	Next, we first prove the following technical claim.
	
	\begin{restatable}{claim}{wplus} \label{cl:wplusinA}
		Let $j \in H$.
		For all $p \in W^{+}_{j, {\sf far}}$, it holds that $u(p) \le \frac{900k \d(p, x_j)}{\varepsilon}$. Therefore,
		$W^{+}_{j, {\sf far}} \subseteq A$. 
	\end{restatable}
	\begin{claimproof}
		For this proof, we use the following shorthand: $W^+ \coloneqq W^{+}_{j, {\sf far}}$, and $W^- \coloneqq W^{-}_{j, {\sf far}}$. Now, fix an arbitrary point $p \in W^+$. For any $q \in W^-$, $\d(p, x_j) \ge \d(q, x_j)$, which implies that 
		\begin{equation}
			\d(p, q) \le \d(p, x_j) + \d(q, x_j) \le 2 \cdot \d(p, x_j) \le 2(1+\tfrac{\varepsilon}{5}) \cdot \d_{r'}(p, x_j) \le 4 \cdot \d_{r'}(p, x_j) \label{eqn:pqbound}
		\end{equation}
		Here, we use the property 4 from \Cref{obs:distances} in the third inequality in the above. On the other hand, note that,
		\begin{align}
			\frac{\varepsilon\,\OPT_r}{200k} \le \frac{\varepsilon\,C_P}{200k} \le C_{W^-}  = \sum_{q \in W^-} \d_{r'}(q, x_j) \le \d_{r'}(p, x_j) \cdot |W^{-}| \label{eqn:wminuscontri}
		\end{align}
		Then, if we set $\alpha = 12 \d_{r'}(p, x_j)$, from (\ref{eqn:pqbound}), we obtain that $W^{-} \subseteq \ball(p, \alpha/3)$. Now, combining this with (\ref{eqn:wminuscontri}), we obtain that,
		\begin{align}
			|\ball(p, \alpha/3)| \ge |W^{-}| \ge \frac{\varepsilon~\OPT_r}{300k \cdot \d_{r'}(p, x_j)} \ge \frac{\varepsilon~\OPT_r}{25k \alpha} = \frac{\varepsilon~\OPT_r}{75k (\alpha/3)} \label{eqn:ballcontain}
		\end{align}
		Hence, we have that $|\ball(p, \nicefrac{75k\alpha}{3\varepsilon})| \ge \frac{\varepsilon~\OPT_r}{75k (\alpha/3)}$.
		Using $\d(p,x_j) \ge \nicefrac{8r}{\varepsilon}$ from Observation~\ref{obs:distances} and the fact $\d_r(p,x_j) \ge 5\d(p,x_j)/6$, we have $\frac{75k\alpha}{3\varepsilon} = \frac{300k\,\d_{r'}(p,x_j)}{\varepsilon} >r$.  Therefore, $u(p) \le  \frac{900k\,\d_{r'}(p,x_j)}{\varepsilon}\le  \frac{900k\,\d(p,x_j)}{\varepsilon}$, as desired.
	\end{claimproof}

	Recall from (\ref{eqn:mostcontri}) that $C_{W_{\sf far} \cap A} \ge \sum_{j \in H} C_{W^+_{j, {\sf far}}} \ge \frac{\varepsilon\,C_P}{8} \ge \frac{\varepsilon\,C_A}{8}$, therefore, when we sample a point from $A$, the probability that the sampled point $p$ belongs to $\bigcup_{j \in H} W^+_{j, {\sf far}}$ is at least $\frac{\varepsilon}{8}$. Now, we condition on this event. Let $o_i \in O$ be the nearest center to $p$, and in line \ref{lin:clusterindex}, the probability that we sample the correct index is $\frac{1}{k}$. Thus, the total probability of the algorithm making the ``correct choices'' in this case is at least $\frac{1}{2} \cdot \frac{\varepsilon}{8} \cdot \frac{1}{k} = \frac{\varepsilon}{16k}$. We condition on these choices. Note that, item 1 is thus satisfied due to the correct sampling of $i$, item 2 is satisfied due $p \in \bigcup_{j \in H} W^+_{j, {\sf far}} \subseteq W_{\sf far} \subseteq W$, and item 3 is satisfied since $p \in A$ by construction. Thus, we are only left with showing item 4, to which end, we consider different cases for the distance between $p$ and its closest optimal center, $o_i$.
	
	If $\d(p, o_i) < \frac{6r}{\varepsilon}$, then since $\d(p, X) \ge \frac{8r}{\varepsilon}$, it easily follows that, 
	\begin{equation}
		\d(p, o_i) < \tfrac{3}{4} \cdot \d(p, X) \le \frac{\d(p, X)}{1+\sfrac{\varepsilon}{12}} \label{eqn:case21}
	\end{equation}	
	
	Otherwise, if $\d(p, o_i) \ge \frac{6r}{\varepsilon}$, then similar item 4 of \Cref{obs:distances}, we can show that, 
	\begin{equation}
		\d_r(p, o_i) < \d(p, o_i) \le (1+\tfrac{\varepsilon}{3}) \cdot \d_r(p, o_i)  \label{eqn:obound}
	\end{equation}
	However, since $p$ is an $\varepsilon/3$-witness, it follows that, $\d_{r'}(p, X) > (1+\frac{\varepsilon}{3}) \cdot \d_r(p, O)$. Therefore, we obtain that,
	\begin{align*}
		\d(p, x_j) \ge \d_{r'}(p, x_j) > (1+\tfrac{\varepsilon}{3}) \cdot \d_r(p, o_i) &\ge \frac{1+\sfrac{\varepsilon}{3}}{1+\sfrac{\varepsilon}{5}} \cdot \d(p, o_i) \ge (1+\tfrac{\varepsilon}{12}) \cdot \d(p, o_i) \numberthis \label{eqn:case22}
	\end{align*}
	Thus, in each of the sub-cases, in (\ref{eqn:case21}) and (\ref{eqn:case22}), we have shown that $\d(p, o_i) < \frac{\d(p, X)}{1+\sfrac{\varepsilon}{12}}$, thus completing the proof of the lemma.
\end{proof}

\Cref{lem:nearbywitnesslemma} and \Cref{lem:farawaywitnesslemma}
can be combined in a straightforward way to obtain the following lemma completing the inductive step.

\begin{restatable}{lemma}{consistencylemma} \label{lem:consistency}
	Consider an iteration of the algorithm such that the state of execution $(X, Q_1, \ldots, Q_k)$ at the start of the iteration is consistent with a fixed optimal solution $O$, and further $\cost_{r'}(P, X) > (1+\varepsilon) \cdot \mathcal{G}$. Then, with probability at least $\frac{\varepsilon}{200k}$, the state of the algorithm at the end of the iteration $(X', Q'_1, Q'_2, \ldots, Q'_k)$ is also consistent with $O$.
\end{restatable}
\begin{proof}
	Consider the scenario described in the premise of the lemma. Then, the solution $X$ at the start of the iteration either satisfies that $C_{W_{\sf near}} \ge \frac{\varepsilon}{100} C_P$ (nearby witness case), or $C_{W_{\sf near}} \ge \frac{\varepsilon}{100} C_P$ (faraway witness case). Then, Lemmas~\ref{lem:nearbywitnesslemma} and \ref{lem:farawaywitnesslemma};
	along with the correctness of \ballintalg, together imply that, conditioned on the respective case assumption, the probability that the state of the algorithm is consistent with $O$ is at least $\frac{\varepsilon}{200k}$. \footnote{Note that this probability also accounts for the result of the coin toss in line \ref{lin:cointoss}.}
\end{proof}
Armed with \Cref{lem:consistency}, it is easy to conclude the proof of \Cref{lem:successprob}.
\begin{proof}[Proof of~\Cref{lem:successprob}]
	As argued initially, the state of the algorithm at line \ref{lin:initsolution} is consistent with the optimal solution $O$. Then, \Cref{lem:consistency} implies that, for any $\ell \ge 1$,
	if the algorithm runs for $\ell$ iterations, then with probability at least $\lr{\frac{\varepsilon}{200k}}^\ell$, the state of the algorithm remains consistent. Further, note that as long as the state of the algorithm remains consistent with $O$, then the algorithm cannot fail.
	Finally, \Cref{lem:runtime} implies that the algorithm terminates within $O(\frac{k}{\varepsilon} \lambda(\frac{\varepsilon}{40}) \log(\frac{k}{40}))$ iterations. 
	Therefore, the probability that the algorithm terminates without failure---i.e., by successfully breaking the \textbf{while} loop by finding a solution $X$ satisfying $\cost_{r'}(P, X) \le (1+\varepsilon)\OPT_r$, is at least $\lr{\frac{\varepsilon}{k}}^{O(\frac{k}{\varepsilon} \lambda(\frac{\varepsilon}{40}) \log(\frac{k}{40}))} = \exp\lr{ -O\lr{\frac{k}{\varepsilon} \log(\frac{k}{\varepsilon}) \lambda(\frac{\varepsilon}{40}) }}$. 
\end{proof}

\section{Extensions of the Bicriteria {\sf FPT-AS} } \label{sec:extensions}

Fomin et al.~\cite{FominGISZ24} defined a generalization of \probname, which they call {\sc Hybrid $(k, z)$-clustering}, wherein the objective is $\sum_{p \in P} (\d_r(p, X))^z$, for fixed $z \ge 1$, which simultaneously generalizes {\sc $(k, z)$-Clustering} and {\sc $k$-Center}.  They generalized their algorithm for $z = 2$, i.e., {\sc Hybrid $k$-Means}, but left the possibility of extending to an arbitrary value of $z$ conditional on the existence of a certain kind of sampling algorithm.
In this paper, we consider a much more general {\sc Hybrid Norm $k$-Clustering} problem that captures all {\sc $(k, z)$-Clustering} problems, and many other advanced problems. Here, the objective is to find a set $X \subseteq \fac$ that minimizes $f(\dd_r(P, X))$, where $\dd_r(P, X) = (\d_r(p_1, X), \d_r(p_2, X), \ldots, \d_r(p_n, X))$ is the vector of $r$-distances to the solution $X$, and $f: \mathbb{R}^n \to \mathbb{R}$ is a \emph{monotone norm}\footnote{$f: \mathbb{R}^n \to \mathbb{R}$ is  a \emph{norm} if it satisfies following properties for any $\mathbf{x}, \mathbf{y} \in \mathbb{R}^n$ and any $\lambda \in \mathbf{R}$, (i) $f(\mathbf{x}) = 0$ iff $\mathbf{x} = \mathbf{0}$, (ii) $f(\mathbf{x}+\mathbf{y}) \le f(\mathbf{x})+f(\mathbf{y})$, (iii) $f(\lambda \mathbf{x}) = |\lambda| f(\mathbf{x})$; furthermore, $f$ is \emph{monotone} if $f(\mathbf{x}) \le f(\mathbf{y})$ whenever $\mathbf{x} \le \mathbf{y}$, where $\le$ denotes point-wise inequality.}---we refer the reader to \cite{AbbasiClustering23} for a detailed background on norms and related notions. 

\Cref{algo:apx} and its analysis can be extended to the general norm objective following the approach of \cite{AbbasiClustering23}. Here, we only sketch the modifications required to the algorithm and its proof, since this is not the main focus of this work.

The initial upper bounds $u(p)$ can be found in a similar manner. The while loop runs as long as the solution $X$ in hand satisfies that $f(\mathbf{d}_{r'}(P, X)) > (1+\sfrac{\varepsilon}{3}) \cdot \OPT_r$, where $\mathbf{d}_{r'}$ denotes the vector of distances $(\d_{r'}(p, X))_{p \in X}$, and $r' = (1+\sfrac{\varepsilon}{3})r$. If this is the case, then the first step is to compute\footnote{Actually, it is sufficient to compute an approximate subgradient, as done in~\cite{AbbasiClustering23}.} a \emph{subgradient} of $f$ at $\mathbf{d}_{r'}$, which is, loosely speaking, a weight function $w: P \to \mathbb{R}_{\ge 0}$, such that it suffices to focus the attention in the current iteration on $\wcost_{r'}(P, X) \coloneqq \sum_{p \in P} w(p) \d_r(p, X)$ that satisfies $\wcost_{r'}(P, X) = f(\mathbf{d}_{r'}(P, X)) > (1+\sfrac{\varepsilon}{3}) \cdot \OPT_r$. Our algorithm proceeds in a similar fashion, except that whenever points are sampled from the set $N$ (in line \ref{lin:sample1}) or set $A$ (in line \ref{lin:sample2}), the probability of sampling a point $p$ is proportional to $w(p) \cdot \d_r(p, X)$. The rest of the algorithm remains unchanged. 

There are a few minor modifications required in the analysis --- mainly in \Cref{ss:successprob}. First, while the definition of an $(\varepsilon/3)$-witness (or simply a \emph{witness}) remains unchanged, we redefine the \emph{contribution} of a subset $S \subseteq P$ to be $C_S \coloneqq \sum_{p \in S} w(p) \cdot \d_r(p, X)$. The nearby and faraway witness cases are then considered exactly as in \Cref{ss:successprob}. The proof of the analogue of \Cref{lem:nearbywitnesslemma} goes through without any modifications; whereas we need to be slightly more careful in the proof of the analogue of \Cref{lem:farawaywitnesslemma}. Here, in inequalities (\ref{eqn:wminuscontri} and \ref{eqn:ballcontain}) we need to take the weighted distances into account and relate $C_P$ to $f(\dd_{r'}(P, X))$, which requires a slight worsening of constants. With these minor changes in place, we can conclude that the state of the algorithm after $\ell$ iterations is consistent with a fixed optimal solution $O$ with probability at least $\Omega(\lr{\frac{\varepsilon}{k}})^\ell$. The analysis for bounding the number of iterations remains unchanged, and hence we obtain an \fptas with a similar running time, modulo the constants hidden in the big-Oh notation. We omit the details.

\section{Coreset}
In this section, we design a coreset for \probname in doubling metric of bounded dimension. More specifically, we prove the following (restated for convenience). 
\coresettheorem*

\iffull
\begin{algorithm}
	\caption{Coreset for \probname}
	\label{algo:coreset}
	\begin{algorithmic}[1]
		\Statex \textbf{Input:} Instance $\cI= (\cli, \fac, \d, r)$ of \probname, $\varepsilon \in (0, 1)$, and a set $T \subseteq \cli \cup \fac, |T|\le \gamma k$ such that $\cost_r(P,T) \le \alpha\cdot \OPT_r$
		\Statex \textbf{Output:} Coreset $(P',\{w(p')\}_{p'\in P'})$ for $\cI$
		\State $P' \gets \emptyset$
		\State Let $R = \frac{\cost_r(P,T)}{\alpha n}$ \Comment{using $\beta=1$}
		\For{$t_i \in T$}
		\For{$j \in \{0,1,\dots, 2\log \lceil \alpha n\rceil\}$}
		\State  let $B_i^j = \ball(t_i,2^jR)$ 
		\EndFor
		\EndFor
		\State Decompose each ball $B^j_i$ into balls of radius each $\frac{\varepsilon 2^jR}{4\alpha}$
		\State Associate each point to a smallest ball containing it (breaking ties arbitrarily)
		\For{$i \in [\gamma k]$}
		\For{$j \in \{0,1,\dots, 2\log \lceil \alpha n\rceil\}$}
		\For{each smaller ball $b$ of $B^j_i$ that has an associated point $p$}
		\State $P' \gets P' \cup \{p\}$
		\State $w(p) =$ number of points in $P$ associated with $b$  
		\EndFor
		\EndFor
		\EndFor
		\State\Return $(P',\{w(p')\}_{p'\in P'})$
	\end{algorithmic}
\end{algorithm}
\fi
\begin{proof} A formal description of our algorithm can be found in \Cref{algo:coreset}, which is based on grid construction approach of~\cite{DBLP:conf/stoc/Har-PeledM04,AbbasiBBCGKMSS24}.
	For each $t_i \in T$, consider the balls $B^j_i = \ball(t_i,2^jR)$, for $j \in \{0,1,\dots, 2\log \lceil \alpha n\rceil\}$. Note that, since $\d(p,T) \le \cost(P,T) = \alpha n \cdot R$, we have that $p$ lies in some ball $B^j_i$, for  $i \in [\gamma k]$ and $j \in  \{0,1,\dots, 2\log \lceil \alpha n\rceil\}$. 
	
	The idea is to decompose each $B^j_i$ into smaller balls each of radius $\frac{\varepsilon2^jR}{4\alpha}$, and associate each point $p \in P$ to a unique smallest ball containing $p$, breaking ties arbitrary. We say a small ball $b$ is non-empty if there is a point in $P$ associated with $b$. Next, for each non-empty ball $b$, pick an arbitrary point $p$ associated with $b$ and add $p$ to $P'$ with weight $w(p)$ equal to the total number of points associated with $b$; we call such a point $p$ as the representative of points associated with $b$. 
	
	To bound the size of $P'$, note that, in doubling metric of dimension $d$, a unit ball can be covered by $2^{O(d\log(1/\varepsilon))}$ balls, each of radius $\varepsilon$. Hence, we have $|P'|=O(2^{O(d\log(1/\varepsilon))} \gamma k\log(\alpha n))$. 
	
	Recall that, for $X\subseteq \fac, |X|\le k$, we define the weighted cost of $X$ with respect to $(P',\{w(p')\}_{p'\in P'})$ as $\wcost_r(P',X) = \sum_{p'\in P}w(p')d_r(p',X)$.
	Now we show that the cost of $X$ with respect to $(P',\{w(p')\}_{p'\in P'})$ approximately preserves the cost of $X$ with respect to $P$.
	Consider a point $p\in P$ and let $p'\in P'$ be its representative in $P'$. Now, the contribution of $p$ towards $\cost_r(P,X)$ is $\d_r(p,X)$. On the other hand, its contribution towards $\wcost_r(P',X)$ is $\d_r(p',X)$. Hence,
	\[
	|\wcost_r(P',X) - \cost_r(P,X)| \le \sum_{p \in P} |\d_r(p,X)- \d_r(p',X)|.
	\]
	Now if $p \in \bigcup_{i \in \gamma k}B^0_i$, then $\d(p,X) - \frac{\varepsilon R}{4\alpha} \le \d(p',X) \le \d(p,X) + \frac{\varepsilon R}{4\alpha}$.
	Hence, 
	\[
	\max\{\d(p,X) - \nicefrac{\varepsilon R}{4\alpha}-r,0\} \le \d_r(p',X) \le \max\{\d(p,X) + \nicefrac{\varepsilon R}{4\alpha}-r,0\}
	\]
	Therefore, $
	\d_r(p,X)- \nicefrac{\varepsilon R}{4\alpha} \le \d_r(p',X) \le \d_r(p,X)+ \nicefrac{\varepsilon R}{4\alpha}$.
	Hence, we have
	\[
	\sum_{p \in P_0} |\d_r(p,X)- \d_r(p',X)| \le \frac{\varepsilon R n}{4\alpha} \le \frac{\varepsilon\OPT_r}{2} \le \frac{\varepsilon\cost_r(P,X)}{2}.
	\]
	
	Now, suppose $p \in P \setminus \bigcup_{i \in \gamma k}B^0_i$, then $\d(p,T) > R$. Let $j\ge 1$ be such that $2^{j-1}R \le \d(p,T) \le 2^jR$. Hence, we have that $2^jR \le 2\d(p,T)$. Therefore, using $\d(p,X) - \frac{\varepsilon d(p,T)}{2\alpha} \le \d(p',X) \le \d(p,X) + \frac{\varepsilon d(p,T)}{2\alpha}$, we have
	\[
	\max\{\d(p,X) - \nicefrac{\varepsilon d(p,T)}{2\alpha}-r,0\} \le \d_r(p',X) \le \max\{\d(p,X) + \nicefrac{\varepsilon d(p,T)}{2\alpha}-r,0\}
	\]
	This implies,$\d_r(p,X)- \nicefrac{\varepsilon d(p,T)}{2\alpha} \le \d_r(p',X) \le \d_r(p,X)+ \nicefrac{\varepsilon d(p,T)}{2\alpha}$.
	Thus, we have
	\[
	\sum_{p \in P \setminus P_0} |\d_r(p,X)- \d_r(p',X)| \le \sum_{p \in P \setminus P_0} \frac{\varepsilon}{2\alpha}\d(p,T) \le \frac{\OPT_r}{2} \le \frac{\varepsilon\cost_r(P,X)}{2}.
	\]
	Finally, we have $|\wcost_r(P',X) - \cost_r(P,X)| \le \varepsilon \cost_r(P,X)
	$, as desired. 
	
	To finish the proof, we invoke Algorithm~\ref{algo:coreset} using set $T$ obtained from the following lemma. At a high level, to obtain this lemma, we start from an $(18, 6)$-bicriteria approximation for \probname from \cite{FominGISZ24,ChakrabartyGK16}, and use the fact that, a ball of radius $O(r)$ can be decomposed into $2^{O(d)}$ balls of radius $r$, converting the guarantee from $\d_{6r}(\cdot, \cdot)$ to $\d_r(\cdot, \cdot)$. 
	\begin{restatable}[Set $T$ for \Cref{algo:coreset}]{lemma}{tric}\label{lem:tric}
		There is a polynomial-time algorithm that, given an instance $\cI= (\cli, \fac, \d, r)$ of \probname in doubling metric of dimension $d$, computes $T \subseteq \cli \cup \fac, |T|\le 2^{O(d)}k$ such that $\cost_r(P,T) \le 36\cdot \OPT_r$, where $\OPT_r$ is the optimal cost for $\cI$.
	\end{restatable}
\end{proof}
\begin{proof}
	Let $A$ be $(18,6)$-bicriteria solution for $\cI$ obtained from~\cite{ChakrabartyS18}, as mentioned in~\cite{FominGISZ24}.
	That is $\cost_{6r}(P,A) \le 18 \OPT_r$. Consider $B_a:=\ball(a,12r)$ for $a\in A$, and decompose $B_a$ into $2^{O(d)}$ smaller balls, each of radius $r/2$ ---note that this follows from the fact that the metric space has doubling dimension $d$. For each smaller ball $b$ such that $b \cap (\cli \cup \fac) \neq \emptyset$, add an arbitrary $t \in b \cap (\cli \cup \fac) $ to $T$. Finally, add $A$ to $T$. Hence, $|T|= |A| + 2^{O(d)}|A| = k2^{O(d)}$.
	
	It is easy to see that, for every $p \in \bigcup_{c \in A} \ball(c, 12r)$, there is some $q \in T$, such that $\d(p, q) \le r$. Now, consider some $p' \not\in \bigcup_{c \in A} \ball(c, 12r)$. Note that $\d_{6r}(p, A) \ge 6r$, which implies that $\d(p, A) \le 2 \cdot \d_{6r}(p, A)$. Therefore, $\d_{r}(p, T) \le \d_r(p, A) \le \d(p, A) \le 2 \d_{6r}(p, A)$. Therefore, it follows that for all $p \in \cli$, $\d_r(p, T) \le 2 \cdot \d_{6r}(p, A)$. Which implies that,
	$$\cost_r(p, T) \le 2 \cdot \cost_r(p, A) \le 36 \cdot \OPT_r.$$
	This concludes the proof of the lemma. 
\end{proof}

\section{Conclusion} \label{sec:conclusion}
In this paper, we revisit \probname, which was introduced and studied recently by Fomin et al.~\cite{FominGISZ24}. We resolve two open questions explicitly asked in their work, namely, for continuous Euclidean instances from $\real^d$, (1) we obtain an \fptas for \probname that does not have an \fpt\ dependence on the dimension $d$ (and in fact, the dependence on $k$ is $2^{O(k \log k)}$ instead $2^{\poly(k)}$ as in \cite{FominGISZ24}), and (2) we design coresets for the same. Indeed, our technique also generalizes to the $(k, z)$-clustering variant of the problem, which was implicitly considered in \cite{FominGISZ24}, but was not completely resolved. To obtain our algorithmic result, we build upon  insights from the recent framework of Abbasi et al.~\cite{AbbasiClustering23} for clustering in metric spaces of bounded \emph{algorithmic scatter dimension} that encapsulates a broad class of metric spaces. Thus, our result shows that the potential of the framework introduced in \cite{AbbasiClustering23} extends to clustering problems beyond that captured by the monotone norm setting, thus paving the way for obtaining {\sf FPT-ASes} for other clustering problems using the similar technique. 
However, several interesting questions remain.

Firstly, the framework of \cite{AbbasiClustering23} is inherently randomized due to key sampling steps, and its derandomization remains an intriguing open question. In particular, derandomizing the step that samples a witness in our algorithm is an interesting challenge.
In another direction, now that the approximation landscape of the vanilla version of \probname is  beginning to be well-understood, it is natural to explore \emph{constrained} variants of the problem such as those involving fairness or capacity constraints. 

\bibliography{reference}
\end{document}